\documentclass[a4paper]{article}
\usepackage[utf8]{inputenc}
\usepackage[english]{babel}

\usepackage{amsfonts}
\usepackage{amsmath}
\usepackage{amssymb}
\usepackage{amsthm}
\usepackage{authblk}

\usepackage[square,sort,comma,numbers]{natbib}

\usepackage{xcolor}
\usepackage{graphicx}
\usepackage{caption}
\usepackage{subcaption}
\usepackage{mathabx}
\usepackage{enumerate}
\usepackage{enumitem} 
\usepackage[title]{appendix} 

\usepackage{setspace}

\usepackage{geometry}
\theoremstyle{break}

\newtheorem{theorem}{Theorem}
\newtheorem{proposition}{Proposition}
\newtheorem{corollary}{Corollary}
\newtheorem{lemma}{Lemma}

\newcommand{\rank}[1]{\mathrm{rank}(#1)}

\renewcommand{\flat}{\mathrm{flat}} 
\newcommand{\hrho}{\hat{\rho}}

\newcommand{\mat}[1]{\left[\begin{matrix}#1\end{matrix}\right]}

\newcommand{\R}{{\mbox{\tiny$\mathtt R$}}}
\newcommand{\Y}{\mbox{\tiny$\mathtt Y$}}

\makeatletter
\def\moverlay{\mathpalette\mov@rlay}
\def\mov@rlay#1#2{\leavevmode\vtop{%
   \baselineskip\z@skip \lineskiplimit-\maxdimen
   \ialign{\hfil$\m@th#1##$\hfil\cr#2\crcr}}}
\newcommand{\charfusion}[3][\mathord]{
    #1{\ifx#1\mathop\vphantom{#2}\fi
        \mathpalette\mov@rlay{#2\cr#3}
      }
    \ifx#1\mathop\expandafter\displaylimits\fi}
\makeatother

\onehalfspace

\title{Parsimony and the rank of a flattening matrix}

\author[1]{Jandre Snyman}
\author[2]{Colin Fox}
\author[3]{David Bryant}
\affil[1]{\small Department of Mathematics and Statistics, University of Otago, Dunedin, New Zealand.}
\affil[2]{\small Department of Physics, University of Otago, Dunedin, New Zealand.}
\affil[3]{\small Department of Mathematics and Statistics, University of Otago, Dunedin, New Zealand.  email: {\tt david.bryant@otago.ac.nz} phone: {\tt +64 34797889}. (Corresponding author)}

\date{\today}

\begin{document}
\maketitle
\begin{abstract}
The standard models of sequence evolution on a tree determine probabilities for every character or site pattern. A flattening is an arrangement of these probabilities into a matrix, with rows corresponding to all possible site patterns for one set $A$ of taxa and columns corresponding to all site patterns for another set $B$ of taxa. Flattenings have been used to prove difficult results relating to phylogenetic invariants and consistency and also form the basis of several methods of phylogenetic inference. We prove that the rank of the flattening equals $r^{\ell_T(A|B)}$, when $T$ is binary, $r$ is the number of states and $\ell_T(A|B)$ is the parsimony length of the binary character separating $A$ and $B$. A similar formula holds for non-binary trees. These results correct an earlier published formula.\\

\noindent Since completing this work, we have learnt that an equivalent result has been proved much earlier by Casanellas and Fern\'andez-S\'anchez \cite{CasanellasFernandez-Sanchez11}, using a completely different proof strategy.\end{abstract}

\section{Introduction}

Behind any statistical inference in phylogenetics is a model of describing the evolution of the states (alleles/nucleotides/amino acids) observed at each site in the alignment. Under the standard model, this evolution is determined by three types of parameters: the phylogeny itself, the distribution of the state at the root and the transition probabilities along each edge. Together these generate the joint distribution for the state at each leaf, which in turn corresponds to a column  of the alignment (reviewed in \cite{BryantGaltierEtal05,Felsenstein04}).

\begin{figure}[ht]
\centerline{\includegraphics[width=0.7\textwidth]{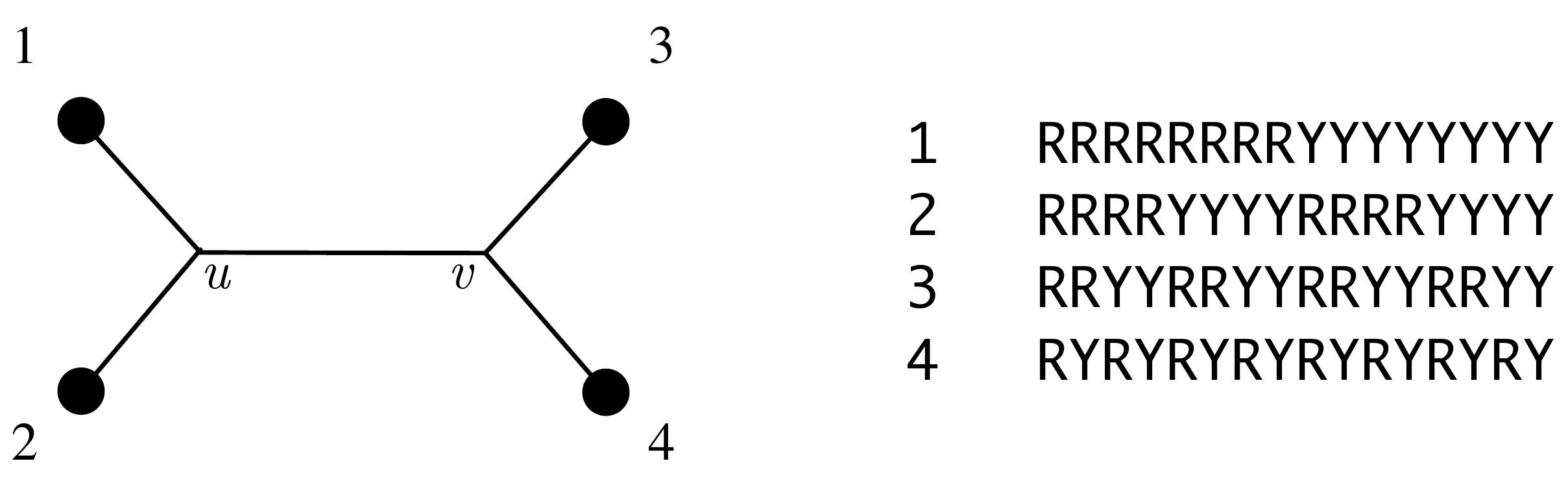}}
\caption{\label{firstFlat} An example of patterns or sites on a four taxa tree. The 16 possible patterns are listed in the table on the right - any one of these could be the pattern for a particular site.}
\end{figure}

An assignment of states to the taxa is often called a {\em pattern}. If there are $r$ states (say $r=2$ for binary, $r=4$ for nucleotide, $r=20$ for amino acids) and $n$ taxa then there are $r^n$ possible patterns. In the example in Figure~\ref{firstFlat} we have $r=2$ and $n=4$, so $r^n = 16$ possible pattern for the joint distribution. Let $p_{i_1i_2i_3i_4}$ denote the joint probability of pattern $i_1i_2i_3i_4$, that is the probability under the model of observing state $i_1$ at leaf $1$, $i_2$ at leaf $2$ and so on. We can think of these as elements of an $r^n$ dimensional vector
\begin{equation}
\mat{p_{\R\R\R\R}\\p_{\R\R\R\Y}\\p_{\R\R\Y\R}\\ \vdots \\ p_{\Y\Y\Y\Y}}\label{pvector}
\end{equation}
or indeed as a $2 \times 2 \times 2 \times 2$ multidimensional array or tensor. Alternatively we can reshape the vector into a matrix, such as 
\begin{equation}
 \mat{p_{\R\R\R\R} & p_{\R\R\R\Y} & p_{\R\R\Y\R} & p_{\R\R\Y\Y} \\
p_{\R\Y\R\R} & p_{\R\Y\R\Y} & p_{\R\Y\Y\R} & p_{\R\Y\Y\Y} \\
p_{\Y\R\R\R} & p_{\Y\R\R\Y} & p_{\Y\R\Y\R} & p_{\Y\R\Y\Y} \\
p_{\Y\Y\R\R} & p_{\Y\Y\R\Y} & p_{\Y\Y\Y\R} & p_{\Y\Y\Y\Y} }.\label{flat4i}
\end{equation}
In this matrix, the rows correspond to the $r^2$ possible ways of assigning states to the first two taxa while the columns correspond to the $r^2$ possible ways of assigning states to the third and fourth taxa.  A matrix of this form is called a {\em flattening}. In tensor terminology, the flattening is an example of an {\em unfolding} of the tensor of pattern probabilities. The idea was introduced into phylogenetics by Pachter and Sturmfels \cite{PachterSturmfels04} and developed extensively by Allman and Rhodes to solve a wide range of mathematical problems in phylogenetics  \cite{RhodesSullivant12,AllmanJarvisEtal13,AllmanRhodes03,AllmanRhodes05,AllmanRhodes06,AllmanRhodes07,AllmanRhodes08,AllmanRhodes09,AllmanRhodesEtal14}.

We can construct a flattening for any partition $A|B$ of the set of taxa into two non-empty parts. The rows of the flattening are indexed by all $r^{|A|}$ ways of assigning a state to the taxa in $A$ and the columns are indexed by all $r^{|B|}$ ways of assigning a state to the taxa in $B$. Each entry equals a term $p_{i_1i_2\cdots i_n}$ with each state $i_k$ determined by the row if $k \in A$ and by the column if $k \in B$. We denote this matrix $\flat_{A|B}$. The matrix in \eqref{flat4i} is $\flat_{\{1,2\}|\{3,4\}}$ and corresponds to the split $\{1,2\} | \{3,4\}$. The flattening for split $\{1,3\}|\{2,4\}$ is 
\[ \flat_{\{1,3\}|\{2,4\}}  = \mat{p_{\R\R\R\R} & p_{\R\R\R\Y} & p_{\R\Y\R\R} & p_{\R\Y\R\Y} \\
p_{\R\R\Y\R} & p_{\R\R\Y\Y} & p_{\R\Y\Y\R} & p_{\R\Y\Y\Y} \\
p_{\Y\R\R\R} & p_{\Y\R\R\Y} & p_{\Y\Y\R\R} & p_{\Y\Y\R\Y} \\
p_{\Y\R\Y\R} & p_{\Y\R\Y\Y} & p_{\Y\Y\Y\R} & p_{\Y\Y\Y\Y} }\]
and the flattening for $\{1\}|\{1,2,3\}$ is
\[\flat_{\{1\}|\{2,3,4\}} = \mat{p_{\R\R\R\R} & p_{\R\R\R\Y} & p_{\R\R\Y\R} & p_{\R\R\Y\Y} & p_{\R\Y\R\R} & p_{\R\Y\R\Y} & p_{\R\Y\Y\R} & p_{\R\Y\Y\Y} \\
p_{\Y\R\R\R} & p_{\Y\R\R\Y} & p_{\Y\R\Y\R} & p_{\Y\R\Y\Y} & p_{\Y\Y\R\R} & p_{\Y\Y\R\Y} & p_{\Y\Y\Y\R} & p_{\Y\Y\Y\Y} }.\]

An important property of flattenings from phylogenies is their rank. Suppose that $e$ is an edge in a phylogeny. Removing $e$ partitions the tree, and hence the set of leaves, into two parts, inducing a partition $A|B$ of the set of taxa. We say that $A|B$ is a split of the tree corresponding to edge $e$. Allman and Rhodes \cite{AllmanRhodes07} (Proposition 11) proved that, under minor assumptions, if $A|B$ is a split in the tree then the rank of $\flat_{A|B}$ is at most $r$, while if $A|B$ is {\em not} a split of the tree then the rank of $\flat_{A|B}$ is at least $r^2$. 

Because of this property, flattenings and their ranks have played a prominent role in the mathematics of phylogenetics, particularly with respect to the development and construction of {\em phylogenetic invariants} \cite{AllmanRhodes08,AllmanRhodes05,PachterSturmfels04}. Roughly speaking, a phylogenetic invariant for a tree is a function on vectors of pattern probabilities such as \eqref{pvector} which is zero when the probability distribution comes from that tree and non-zero otherwise. 

Flattenings have also led to methods for inferring phylogenies directly. In an  original and influential chapter, Eriksson \cite{Eriksson05} outlines an efficient method for inferring phylogenies with few assumptions about the evolutionary process. The SVD quartets method \cite{ChifmanKubatko14} uses flattenings to infer trees for subsets of four taxa, subsequently assembling these four-taxa trees into one for the complete set of taxa. The method is statistically consistent even in the presence of incomplete lineage sorting. Quartet-based approaches based on flattening have also been developed by \cite{Fernandez-SanchezCasanellas16,CasanellasFernandez-SanchezEtal21}. Allman, Kubako and Rhodes propose the ranks of flattenings as a measure of support for different edges in a phylogeny  \cite{AllmanKubatkoEtal17}.\\~\\

In this paper we extend and complete the theorem of Allman and Rhodes and derive an exact formula for the rank of arbitrary flattenings of a tree. Our result corrects a formula appearing in \cite{Eriksson05}. We show that the rank of $\flat_{A|B}$ is given by $r^{\ell_T(A|B)}$, where $\ell_T(A|B)$ is the parsimony length of the binary character which assigns one state to all leaves in $A$ and another state to all leaves in $B$. A slightly modified result holds for non-binary (multifurcating) trees. 

In the next section we present the notation and basic results needed to establish this result. We note that we generalise from characters to maps on arbitrary subsets of vertices, and define our concepts in this more general context. The main result is proved in a series of lemmas and propositions in the final section. ~\\

Since completing this work, we have learnt that an equivalent result has been proved much earlier by Casanellas and Fern\'andez-S\'anchez \cite{CasanellasFernandez-Sanchez11}, using a completely different proof strategy.

\section{Background}

\subsection{Trees and characters}

An {\em unrooted phylogeny} is an acyclic, connected graph $T = (V,E)$ with leaf set (taxon set) $L(T)$. We say that $T$ is binary (full resolved) if every non-leaf vertex has degree three.  In a {\em rooted phylogeny} $T = (V,E_\rho)$ one vertex is selected as the root $\rho$ and edges are directed away from $\rho$. The rooted phylogeny $T$ is binary if every non-leaf vertex has out-degree $2$.

A {\em character} is a function $f$ from $L(T)$ to a set of states $[r] = \{1,2,\ldots,r\}$. Hence $r=4$ for nucleotides and $r=20$ for amino acids. We will consider a more general situation where $f$ is a map from any subset of $V$ to $[r]$, not just the set $L(T)$. 

The {\em length} of a function $F:V \rightarrow [r]$ is defined as
\[\ell_T(F) = |\left\{ \{u,v\} \in E: F(u) \neq F(v) \right\}|\]
and the (parsimony) length a  function $f$ with domain $A \subseteq V$ is the length of a minimum extension
\[\ell_T(f) = \min\{ \ell_T(F) : F|_A = f \}.\]

The length of a function or character can be expressed equivalently using vertex and edge cuts. For $E' \subseteq E$ or $V' \subseteq V$ we let $T \setminus  E'$ and $T \setminus V'$ denote the graphs resulting from deleting $E'$ or $V'$ respectively. 
The length $\ell_T(f)$ equals the minimum cardinality of an edge cut  $E' \subseteq E$ such that $f(u) = f(v)$ whenever $u$ and $v$ are in the same component of $T\setminus E'$.  In a similar way, we let $\nu_T(f)$ denote the minimum cardinality of a vertex cut  $V' \subseteq V$ such that 
$f(u) = f(v)$ whenever $u$ and $v$ are in the same component of $T \setminus V'$. For all $f$ we have $\nu_T(f) \leq \ell_T(f)$.

If $A$ and $B$ are disjoint subsets of $V$ then we let $\ell_T(A|B)$ and $\nu_T(A|B)$ denote $\ell_T(f)$ and $\nu_T(f)$, where $f$ is the function given by
\[f(x) = \begin{cases} 0 & \mbox{ if $x\in A$} \\ 1 & \mbox{ if $x \in B$.} \end{cases}\]

Steel \cite{Steel93,SempleSteel03} observed that the parsimony length  of a character with two states can be expressed in terms of the size of disjoint path sets . This applies for both  $\ell_T(A|B)$ and $\nu_T(A|B)$. 

\begin{proposition} \label{Menger}
Let $T = (V,E)$ be an unrooted phylogeny and let $A$ and $B$ be disjoint subsets of $V$. Then $\ell_T(A|B)$ equals the maximum cardinality of a set of  edge-disjoint paths connecting vertices in $A$ with vertices in $B$, while $\nu_T(A|B)$ equals the maximum cardinality of a set of vertex-disjoint paths connecting vertices in $A$ with vertices in $B$. \end{proposition}

\begin{corollary} \label{vertexEdge}
Let $T = (V,E)$ be a binary phylogeny and suppose $A$ and $B$ are disjoint subsets of $L(T)$. Then $\ell_T(A|B) = \nu_T(A|B)$.
\end{corollary}
\begin{proof}
If $p_1$ and $p_2$ are any two paths in a binary phylogeny which begin and end at leaves then $p_1$ and $p_2$ are vertex disjoint if and only if they are edge disjoint.
\end{proof}

\subsection{Stochastic models on trees}

 The standard stochastic models on phylogenies determine the joint probabilities for the states at every vertex of the tree.  
Let $T = (V,E_\rho)$ be a rooted tree with root $\rho$. Let $\pi_\rho$ denote the root distribution  and for each directed edge $(u,v)$ we associate an $r \times r$ transition probability matrix $P_{uv}$. Let $X_v$ denote the random state associated with vertex $v$ and, for $A \subseteq V$, let $X_A$ denote the joint random variable  $X_a:a \in A$.  

The joint probability that  $X_v = F(v)$ for all $v \in V$ is 
\[\pi(X_V = F) = \pi_\rho(F(\rho)) \prod_{(u,v) \in E_\rho} P_{uv} (F(u),F(v)).\]
 The marginal probabilities for maps $f:A \rightarrow [r]$  on some subset $A \subseteq V$  are then given by
\[\pi(X_A = f) = \sum_{F:F|_A = f} \pi(F).\]
With this notation, the standard probability for a character $f:L(T) \rightarrow [r]$ is $\pi(X_{L(T)} = f)$.

We will make three assumptions about the parameters of the stochastic model.
\begin{enumerate}
\item[(C1)] The transition matrices $P_{uv}$ are non-singular.
\item[(C2)]  $\pi_\rho(i) >0$ for all $i \in [r]$.
\item[(C3)] $P_{uv}(i,j)>0$ for all edges $(u,v)$ and $i,j \in [r]$.
\end{enumerate}
We note that all three are satisfied in standard evolutionary models, where mutations along a branch are modelled using continuous time Markov chains.

These conditions have three immediate consequences.
\begin{proposition} \label{PairRank}
Suppose that conditions (C1) and (C2) are satisfied. Then
\begin{enumerate}
\item $\pi(X_v = i) >0$ for all $v \in V$ and $i \in [r]$.
\item For all $u,v \in V$ (not necessarily adjacent) the $r \times r$ matrix with entries $P_{uv} = \pi(X_v = j | X_u = i)$ is non-singular.
\item For all $u,v \in V$ the $r \times r$ matrix with entries $M_{ij} = \pi(X_u = i , X_v = j)$ is non-singular.
\end{enumerate}
\end{proposition}
\begin{proof}
\begin{enumerate}
\item See \cite{SteelMA1994Rtws}, Theorem 2.
\item Let $v_1,v_2,\ldots,v_k$ be the path from $v_1 = u$ to $v_k = v$. Then 
\[P_{uv} = P_{v_1v_2} P_{v_2v_3} \cdots P_{v_{k-1}v_{k}},\]
which is non-singular by (C1).
\item Let $\Pi_u$ be the diagonal matrix with entries $\pi(X_u = i)$ on the diagonal. Then $M = \Pi P_{uv}$ so is non-singular by parts 1 and 2.
\end{enumerate}
\end{proof}

For any edge $(u,v)$ in the rooted tree $T=(V,E_\rho)$ we can define a transition matrix $P_{vu}$ for the reverse edge $(v,u)$ using Bayes' rule:
\begin{align}
P_{vu}(i,j) &=\pi(X_u = j|X_v = i) \nonumber \\
& = \pi(X_v = i|X_u = j) \frac{\pi(X_u=j)}{\pi(X_v=i)} \label{BayesFlip}\\
& = P_{uv}(j,i) \frac{\pi(X_u=j)}{\pi(X_v=i)}. \nonumber
\end{align}
This formula allows us to move the root to any vertex in the tree with no change to the  distribution of joint random variables $X_V$ \cite{AllmanRhodes03,SteelMA1994Rtws}. To move from the root from $\rho$ to $\hrho$ we make the new root distribution $\pi_{\hrho} = \pi(X_{\hrho})$ and use \eqref{BayesFlip} to determine transition matrices for all edges which are flipped when directing away from $\hrho$ instead of from $\rho$ (under the assumption that (C1) and (C2) both hold). 

Conditional independence for the variables $X_v$ are, as with any graphical model, determined by cuts in the graph. More specifically, if $A_1,A_2,\ldots,A_k$ are the components of $T\setminus\{v\}$ then $X_{A_1},X_{A_2},\ldots,X_{A_k}$ are conditionally independent, given $X_v$, see \cite[Lemma 7.1]{Steel16}. The following Proposition is a direct consequence of this result

\begin{proposition} \label{conditional}
If $V' \subseteq V$ and $A_1,A_2,\ldots,A_k$ are the components of $T\setminus\{V'\}$ then $X_{A_1},X_{A_2},\ldots,X_{A_k}$ are conditionally independent, given $X_{V'}$.
\end{proposition}
\begin{proof}
Apply Lemma 7.1 of \cite{Steel16} with respect to some $v \in V'$ and then apply the result recursively on the components.
\end{proof}

\subsection{The main theorem}

Let $T=(V,E)$ be a phylogeny and suppose that $A$ and $B$ are disjoint subsets of $V$. We let $\flat_{A|B}$ denote the matrix with rows indexed by maps $f_A:A \rightarrow [r]$, columns indexed by maps $f_B:B \rightarrow [r]$ and $(f_A,f_B)$ entry given by
\[\flat_{A|B}(f_A,f_B) = \pi(X_A = f_A, X_B = f_B).\]
Hence $\flat_{A|B}$ is an $r^{|A|} \times r^{|B|}$ non-negative matrix with entries which sum to $1$. 

\begin{theorem} \label{MainTheorem}
Let $T = (V,E)$ be a phylogeny and let $A$ and $B$ be disjoint subsets of $V$. Suppose that the root distribution and transition matrices satisfy (C1)---(C3). Then
\[ \rank{\flat_{A|B}} = r^{\nu_T(A|B)}.\]
If $T$ is binary and $A$ and $B$ subsets of $L(T)$ then
\[ \rank{\flat_{A|B}} = r^{\ell_T(A|B)}.\]
\end{theorem}

\section{Proof of Theorem~\ref{MainTheorem}}

We prove the main theorem in two steps. First, we show that $r^{\nu_T(A|B)}$ provides an {\em upper bound} for the rank of $\flat_{A|B}$; second we show that this upper bound is actually achieved.

\begin{lemma} \label{upperBound}
Let $T = (V,E)$ be a phylogeny and suppose that $A$ and $B$ are disjoint subsets of $V$. 
Then 
\[\rank{\flat_{A|B}} \leq r^{\nu_T(A|B)}.\]
\end{lemma}
\begin{proof}
Let $C$ be a minimum cardinality vertex cut of $T$  such that each component of $T \setminus C$ contains vertices from at most one of $A$ or $B$. By Proposition~\ref{conditional}  $X_{A \setminus C}$ and $X_{B \setminus C}$ are conditionally independent given $X_C$. Hence for all maps $f_A:A \rightarrow [r]$ and $f_B:B \rightarrow [r]$ we can factor $\pi(X_A = f_A, X_B = f_B)$ as 
\begin{align*}
\pi(X_A = f_A, X_B = f_B) &= \sum_{f_C} \pi(X_A = f_A, X_B = f_B|X_C = f_C ) \pi(X_C = f_C) \\
& = \sum_{f_C} \pi(X_A = f_A |X_C = f_C)  \pi(X_B = f_B, X_C = f_C) \\
\end{align*}
Define the $r^{|A|} \times r^{|C|}$ matrix $R$ by
\[R(f_A,f_C) = \pi(X_A = f_A |X_C = f_C),\]
and the $r^{|C|} \times r^{|B|}$ matrix $S$ by
\[S(f_C,f_B) = \pi(X_B = f_B, X_C = f_C).\]
Then
\[\flat_{A|B} = RS\]
and 
\[\rank{\flat_{A|B}} \leq \mathrm{rank}(R) \leq r^{|C|} = r^{\nu_T(A|B)}.\]
\end{proof}

To illustrate, consider the tree in figure~\ref{counterExample}. The sets $A$ and $B$ are indicated by filled and unfilled leaves respectively. Removing the three marked vertices separates all leaves in $A$ from leaves in $B$. Hence $\nu_T(A|B) \leq 3$ and $
\rank{\flat_{A|B}} \leq r^{\nu_T(A|B)}$. We note that the formula in \cite{Eriksson05} gives a generic rank of $r^4$ for this flattening.

\begin{figure}[ht]
\centerline{\includegraphics[width=0.6\textwidth]{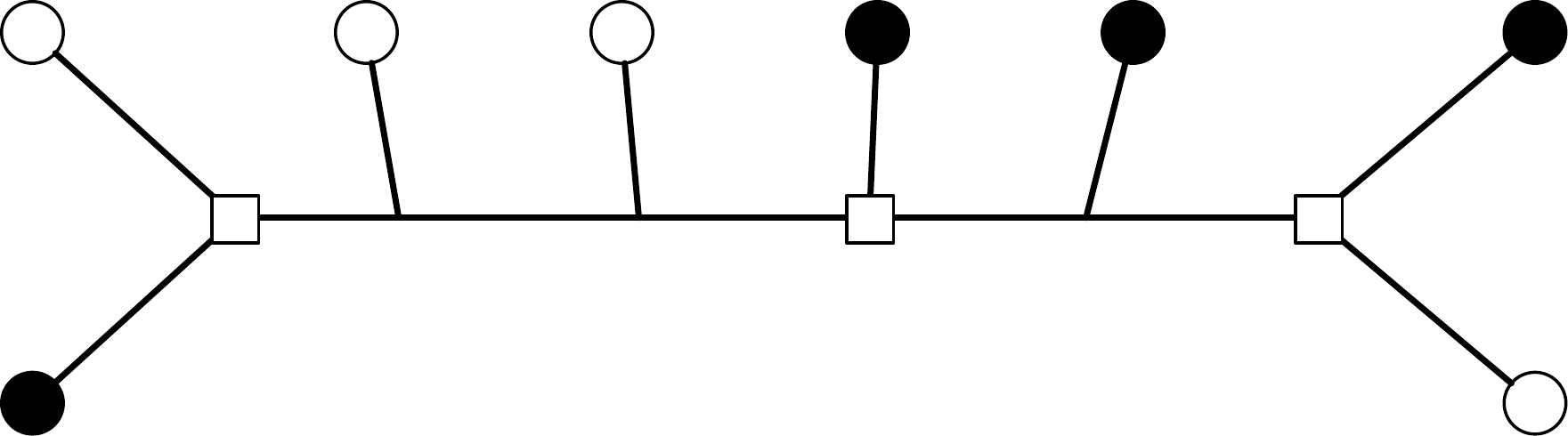}}
\caption{\label{counterExample} A counter-example for the rank formula in \cite{Eriksson05}. Let $A$ and $B$ be given by the filled and unfilled leaf nodes, respectively. The three squares mark out a vertex cut of size $3$ separating $A$ from $B$. Hence the rank of $\flat_{A|B}$ is at most $r^3$, for $r$ states. Theorem 19.5 of \cite{Eriksson05} states that the generic rank is $r^4$.}
\end{figure}

The second part of the proof of the main theorem is proving a matching lower bound for Lemma~\ref{upperBound}. We start by demonstrating that the rank of $\flat_{A|B}$ increases monotonically as a function of $A$ and $B$.

\begin{lemma} \label{monotonic}
Let $A$ and $B$ be disjoint subsets of $V$ and suppose $A' \subseteq A$ and $B' \subseteq B$. Then 
\[\rank{\flat_{A'|B'}} \leq \rank{\flat_{A|B}}.\]
\end{lemma}
\begin{proof}
For any $f_{A'}:A' \rightarrow [r]$ and $f_{B'}:B' \rightarrow [r]$ we have 
\begin{align*}
\flat_{A'|B'}(f_{A'},f_{B'}) & = \pi(X_{A'} = f_{A'}, X_{B'} = f_{B'}) \\
& = \sum_{\substack{g_A:A \rightarrow [r] \\
g_A|_{A'} = f_{A'} }} \sum_{\substack{g_B:B \rightarrow [r] \\
g_B|_{B'} = f_{B'} }} \pi(X_A = g_A, X_B = g_B) \\
& = \sum_{\substack{g_A:A \rightarrow [r] \\
g_A|_{A'} = f_{A'} }} \sum_{\substack{g_B:B \rightarrow [r] \\
g_B|_{B'} = f_{B'} }} \flat_{A|B}(g_A,g_B).
\end{align*}
Hence there is an $r^{|A'|} \times r^{|A|}$ $0\!-\!1$ (binary) matrix $U_A$ and a $r^{|B'|} \times r^{|B|}$ $0\!-\!1$ matrix $U_B$ such that 
\[\flat_{A'|B'} = U_A \flat_{A|B} U_B^T\]
and $\rank{\flat_{A'|B'}} \leq \rank{\flat_{A|B}}$.
\end{proof}

The next Lemma establishes Theorem~\ref{MainTheorem} in the extremal case that $\nu_T(A|B) = |A| = |B|$. This is where the bulk of the work proving the main theorem is carried out. 

\begin{lemma} \label{reduction}
Suppose that $T=(V,E)$, $A$ and $B$ satisfy the conditions of Theorem~\ref{MainTheorem}.
If $|A| = |B| = m$ and there are $m$ disjoint paths connecting elements of $A$ to elements of $B$ then $\rank{\flat_{A|B}} = r^m$.
\end{lemma}
\begin{proof}
We prove the result by induction on $m$. 

For the base case suppose that $m=1$, $A = \{a\}$ and  $B = \{b\}$. Then 
\[\flat_{A|B}(i,j) = \pi(X_a = i,X_B = j)\]
which is full rank $r = r^{\nu_T(A|B)}$ by Proposition~\ref{PairRank}.

Next, assume that the result holds when $|A| = |B| = m$. Suppose that there is a collection of $m+1$ vertex disjoint paths connecting elements of $A'$ and $B'$, where $|A'| = |B'| = m+1$. Fix a vertex $\rho \in V$ and let $p$ be the path furthest away from $\rho$. Suppose that this path goes from $a \in A'$ to $b \in B'$. Let $v \in p$ be the vertex on $p$ closest to $\rho$ and let $e = \{v,w\}$ be the first edge on the path from $v$ to $\rho$. If there is a path $p'$ in the same component of $T \setminus \{e\}$ then $v$ is on the path from $p'$ to $\rho$, contradicting the choice of $p$.

Let $A = A'\setminus \{a\}$ and $B = B' \setminus \{b\}$.  Define the matrices $F,G$ by 
\begin{align*}
F(i,k) &= \pi(X_a = i|X_v = k) \\
G(k,j) &= \pi(X_v = k,X_b = j)
\intertext{and for each $k \in [n]$  define the $r^m\times r^m$ matrix}
H^{(k)}(f,g) & = \pi(X_A = f,X_B=g|X_v = k).
\end{align*}
Note that $F$ and $G$ are both non-singular, by Lemma~\ref{PairRank}. Furthermore, for each $k$ the matrix $H^{(k)}$ equals the flattening matrix $\flat_{A|B}$ but with respect to root $\rho=w$ and  root distribution  
\[\pi_\rho(i) = P_{vw}(k,i).\]
Hence by the induction hypothesis, $H^{(k)}$ also has full rank.

Suppose that $\hat{f}:\hat{A} \rightarrow [r]$ and $\hat{g}:\hat{B} \rightarrow [r]$. Let $i = \hat{f}(a)$,  $j = \hat{g}(b)$  and let $f,g$ be the  restrictions of $\hat{f}$ and $\hat{g}$ to $A$ and $B$. Then 
\begin{align*}
\flat_{\hat{A} |\hat{B}} (\hat{f},\hat{g}) & = \sum_{k=1}^r F(i,k) H^{(k)}(f,g) G(k,j).
\end{align*}
We will suppose that the null space of $\flat_{\hat{A} |\hat{B}} (\hat{f},\hat{g})$ contains no non-zero vectors. To this end, 
let $x$ be a vector indexed by pairs $(g,j)$ with $g:B \rightarrow [r]$ and $j \in [r]$ such that
\begin{align*}
\sum_g \sum_{j=1}^r \sum_{k=1}^r F(i,k) H^{(k)}(f,g) G(k,j) x(g,j) & = 0
\intertext{for all $f:A \rightarrow [r]$ and $i \in [r]$. Rearranging, we have }
\sum_{k=1}^r F(i,k)  \left[ \sum_g \sum_{j=1}^r  H^{(k)}(f,g) G(k,j) x(g,j) \right] & = 0
\intertext{and since $F$ is non-singular, }
\sum_g \sum_{j=1}^r  H^{(k)}(f,g) G(k,j) x(g,j) & = 0
\end{align*}
for all $f$ and $k$. Define the vector $y$ by
\begin{equation}
y(g,k) = \sum_{j=1}^r G(k,j) x(g,j) \label{yform}
\end{equation}
so that for all $k$,
\begin{align*}
\sum_g  H^{(k)}(f,g) y(g,k) & = 0.
\end{align*}
As $H^{(k)}$ is non-singular, $y(g,k) = 0$ for all $g,k$, from which \eqref{yform} and the fact that $G$ is non-singular gives $x = 0$.  We have shown that the null space of $\flat_{\hat{A}|\hat{B}}$ is trivial, proving the lemma.
\end{proof}

We can now prove the main theorem.

\begin{proof} (Theorem~\ref{MainTheorem})
Suppose that $A$ and $B$ are disjoint subsets of $V$. From Lemma~\ref{upperBound} we have \[\rank{\flat_{A|B}} \leq r^{\nu_T(A|B)}.\]

By Proposition~\ref{Menger} there are $\nu_T(A|B)$ vertex disjoint paths connecting vertices in $A$ to vertices in $B$. Let $A'$ and $B'$ be the endpoints of these paths, so $|A'| = |B'| = \nu_T(A'|B') = \nu_T(A|B)$. By Lemma~\ref{reduction} we have $\rank{\flat_{A'|B'}} = r^{\nu_T(A|B)}$ and by Lemma~\ref{monotonic} we have 
\[\rank{\flat_{A|B}} \geq \rank{\flat_{A'|B'}} = r^{\nu_T(A|B)}.\]

The second part of the theorem now follows from Corollary~\ref{vertexEdge}.
\end{proof}

We note that none of the conditions (C1)---(C3) on the root distribution and transition matrices can be eliminated. For example, consider the four taxa tree 
in Figure~\ref{firstFlat}, a case studied in detail by \cite{AllmanRhodes06}. The joint probability distribution for $X_1,X_2,X_3,X_4$ is
\[\pi(X_1=a,X_2=b,X_3=c,X_4=d)  = \sum_i \sum_j P_{u1}(i,a) P_{u2}(i,b) \pi_{\rho}(i) P_{uv}(i,j) P_{v3}(j,c) P_{v4}(j,d) .\]
Gather terms, we obtain a decomposition
\[ \flat_{\{1,3\}|\{2,4\}} = UDV \]
where $U$ and $V$ are $r^2 \times r^2$ matrices and $D$ is a diagonal matrix with diagonal entries
\[D_{ij;ij} = \pi_{\rho}(i) P_{uv}(i,j).\]
From here we see that if there is $i$ such that $\pi_\rho(i) = 0$, or $ij$ such that $P_{uv}(i,j) = 0$, then 
\[\rank{\flat_{\{1,3\}|\{2,4\}} } < r^2.\]


\bibliographystyle{plain}
\bibliography{Refs}

\end{document}